\RequirePackage{amsmath}
\documentclass{llncs}
\usepackage{multisets}
\usepackage{comment,booktabs}

\title{Approximation and Dependence\\ via Multiteam Semantics}

\author{Arnaud Durand\inst{1} \and Miika Hannula\inst{2} \and Juha Kontinen\inst{2} \and Arne Meier\inst{3} \and Jonni~Virtema\inst{3}}

\institute{Institut de Math\'ematiques de Jussieu - Paris Rive Gauche, CNRS UMR 7586 - Universit\'e Paris Diderot, \email{durand@math.univ-paris-diderot.fr}
 \and Department of Mathematics and Statistics, University of Helsinki, \email{$\{$juha.kontinen,miika.hannula$\}$@helsinki.fi} \and Leibniz Universität Hannover, Institut für Theoretische Informatik, \email{$\{$meier,virtema$\}$@thi.uni-hannover.de}}

\begin{document}
 \maketitle

\begin{abstract}
We define a variant of team semantics called \emph{multiteam semantics} based on multisets and study the properties of various logics in this framework. In particular, we define natural probabilistic versions of inclusion and independence atoms and certain approximation operators motivated by approximate dependence atoms of V\"a\"an\"anen. 
\end{abstract}

\section{Introduction}

Dependence logic was introduced by V\"a\"an\"anen in 2007 \cite{vaananen07}. It extends first-order logic with dependence atomic formulas (dependence atoms) $\dep(\tuple x, y)$ with the intuitive meaning that the value of the variable $y$ is functionally determined by the values of the variables $\tuple x$.
The notion of dependence has real meaning only in plurals. 
Thus, in contrast to the usual Tarskian semantics, in dependence logic the satisfaction of formulas is defined not via single assignments but via sets of assignments. Such sets are called \emph{teams} and the semantics is called \emph{team semantics}. In this article we take a further step of replacing structures and teams by their multiset analogues. Multiteams have been considered in some earlier works \cite{2015arXiv150901812H,hpv14,v14} but so far no systematic study of the subject in the team semantics context has appeared. In the temporal logic setting (in the context of computation tree logic) multiteam semantics have been introduced and studied recently \cite{kmv15}. In this article we define the so-called \emph{lax} and \emph{strict} multiteam semantics and study properties of various logics under these semantics. Moreover we show how the shift from sets to multisets naturally gives rise to probabilistic and approximate versions of dependence logic.

The idea of team semantics goes back to Hodges \cite{MR1465612} whose aim was to define compositional semantics for \emph{independence-friendly logic} \cite{hisa89}. The introduction of dependence logic and its many variants has evinced that team semantics is a very interesting and versatile semantical framework. In fact, team semantics has natural propositional, modal, and temporal variants. The study of \emph{modal dependence logic} was initiated by V\"a\"an\"anen \cite{VaananenMDL} in 2008. Shortly after, \emph{extended modal dependence logic} was introduced by Ebbing et~al. \cite{EHMMVV13} and \emph{modal independence logic} by Kontinen et~al. \cite{DBLP:journals/corr/KontinenMSV14a}. In purely propositional context the study was initiated by Yang and V\"a\"an\"anen \cite{DBLP:journals/corr/YangV14} and further studied, e.g., by Hannula et~al. \cite{DBLP:journals/corr/HannulaKVV15}. 
One of the most important developments in the area of team semantics was the introduction of \emph{independence logic} \cite{gradel13} in which dependence atoms of dependence logic are replaced by \emph{independence atoms}
$\tuple y \perp_{\tuple x} \tuple 
z$. 
The intuitive meaning of the independence atom $\tuple y \perp_{\tuple x} \tuple z$ is that, when the value of $ \tuple x$ is fixed, knowing the value of $\tuple z$ does not tell us anything new about the value of $\tuple y$. Soon after the introduction of independence logic, Galliani \cite{Galliani:2011} showed that independence atoms can be further analysed, and alternatively expressed, in terms of inclusion and exclusion atoms.
The inclusion atom $\tuple x \subseteq \tuple y$ expresses that each value taken by $\tuple x$ in a team $X$ appears also as a value of $\tuple y$ in $X$. The meaning of the exclusion atom $ \tuple x | \tuple y$ is that $\tuple x$ and $\tuple y$ have no common values in $X$.

Independence, inclusion, and exclusion atoms have very interesting properties in the team semantics setting. For example, inclusion atoms give rise to a variant of dependence logic that corresponds to the complexity class PTIME over finite ordered structures \cite{gh13}. In fact, the complexity theoretic aspects of these atoms in propositional, modal, and first-order setting have been studied extensively during the past few years (see the survey of Durand et~al.\ \cite{DKV} and the references therein).
 
A team $X$ over variables $x_1,\ldots, x_{n}$ can be viewed as a database table with $x_1,\ldots, x_{n}$ as its attributes. Under this interpretation, dependence, inclusion, exclusion, and independence atoms correspond exactly to functional, inclusion, exclusion, and embedded multivalued dependencies, respectively. These dependencies have been studied extensively in database theory. The close connection between team semantics and database theory has already led to fruitful interactions between these areas \cite{DBLP:conf/foiks/HannulaK14,DBLP:conf/cikm/HannulaKL14,DBLP:conf/wollic/KontinenLV13}. It is worth noting that multiset semantics (also known as bag semantics) is widely used in databases \cite{DBLP:conf/cikm/BottcherLZ14,DBLP:conf/amw/Kolaitis13,Lamperti:2000:MDS:647269.721839}. On the other hand, independence atoms, embedded multivalued dependencies, and the notion of conditional independence $\tuple Y \perp\tuple Z|\tuple X$ in statistics have very interesting connections, see, e.g., \cite{Gyssens2014628,wbw00}. In this article we establish that, in the multiteam semantics setting, independence atoms can be naturally interpreted exactly as statistical conditional independence. Probabilistic versions of dependence logic have been previously studied by Galliani and Mann \cite{Galliani2008,Mann}.

In practice dependencies such as functional dependence do not hold absolutely but with a small margin of error. In order to logically model such scenarios, V\"a\"an\"anen introduced approximate dependence atoms \cite{v14}. The corresponding approximate functional dependencies have been studied in the context of data mining \cite{DBLP:journals/tcs/KivinenM95}. In this article we define a general approximation operator which, in particular, can be used to express approximate dependence atoms. In the last sections of the article, we study the computational aspects of logics extended by the approximation operator.

  \paragraph{Previous work on multisets in team semantics} The idea of generalising team semantics by the use of multisets has been discussed in several articles. Hyttinen et~al.\ \cite{2015arXiv150901812H} study multiteams, and their generalisations called \emph{quantum teams}, which they use to give semantics to a propositional logic called \emph{quantum team logic} that can be used for the logical analysis of phenomena in quantum physics. Moreover Hyttinen et~al.\ \cite{hpv14} define a notion of a \emph{measure team} and \emph{measure team logic}. The latter is a logic for making inferences about probabilities of first-order formulas in measure teams. Furthermore Krebs et~al.\ introduced team semantics with multisets for the temporal logic CTL \cite{kmv15}. Finally the fact that under multiteam semantics approximate dependence atoms have the locality property (compare to Proposition \ref{locality_approx}) is discussed by V\"a\"an\"anen \cite{v14}. 

 \paragraph{Organisation.} This article is organised as follows. Section~\ref{Prel} briefly discusses the basic concepts and definitions. The generalisation of team semantics to multisets is presented in Section~\ref{Mult}. Section~\ref{Appr} defines the approximation operators, and in Section~\ref{Comp} the complexity theoretic aspects of logics with the approximation operators are studied.

\section{Preliminaries}\label{Prel}
We assume familiarity with standard notions in computational complexity theory and logic. We will make use of the complexity classes $\NP$ and $\Ptime$. For an introduction into this topic we refer to the good textbook of Papadimitriou \cite{DBLP:books/daglib/0072413}.

\subsection{Team Semantics}

Vocabularies $\tau$ are finite sets of relation symbols with prescribed arities. For each $R\in\tau$, let $\ar(R)\in Z_+$ denote the arity of $R$. A $\tau$-structure is a tuple $\A = \big(A, (R^\A_i)_{R_i\in\tau}\big)$, where $A$ is a set and each $R^\A_i$ is an $ar(R_i)$-ary relation on $A$ (i.e., $R^\A_i \subseteq A^{ar(R_i)}$). We use $\A$, $\mathfrak{B}$, etc. to denote $\tau$-structures and $A$, $B$, etc.\ to denote the corresponding domains. In this article we restrict attention to finite structures.

Let $D$ be a finite set of first-order variables and $A$ be a nonempty set. A function $s\colon D \to A$ is called an \emph{assignment}. The set $D$ is the \emph{domain} of $s$, and the set $A$ the \emph{codomain} of $s$. For a variable $x$ and $a\in A$, the assignment $s(a/x)\colon D\cup\{x\} \rightarrow A$ is obtained from $s$ as follows:
\[
s(a/x)(y) :=
\begin{cases}
a & \text{if $y=x$},\\
s(y) &\text{otherwise}. 
\end{cases} 
\]

A $\emph{team}$ is a finite set of assignments with a common domain and codomain.
Let $X$ be a team, $A$ a finite set, and $F\colon X\to \Po(A)\setminus \{\emptyset\}$ a function. We denote by $X[A/x]$ the modified team $\{s(a/x) \mid s\in X, a\in A\}$, and by $X[F/x]$ the team $\{s(a/x)\mid s\in X, a\in F(s)\}$. Let $\A$ be a $\tau$-structure and $X$ a team with codomain $A$, then we say that $X$ is a team of $\A$.

Let $\tau$ be a set of relation symbols. The syntax of first-order logic $\FO(\tau)$ is given by the following grammar, where $R\in \tau$, $\tuple x$ is a tuple of variables, and $x$ and $y$ are variables. Note that in the definition the scope of negation is restricted to atomic formulae.
\[
\varphi ::=  x=y \mid x\neq y \mid R(\tuple x) \mid \neg R(\tuple x) \mid (\varphi \land \varphi) \mid (\varphi \lor \varphi) \mid \exists x \varphi \mid \forall x \varphi.
\]
Let $\tuple x,\tuple y$ be tuples of variables and $\varphi$ a formula. We write $\Var(\varphi)$ for the set of variables that occur in $\varphi$, and $\Var(\tuple x)$ for the set of variables listed in $\tuple x$. We also write $\tuple x\tuple y$ for the concatenation of $\tuple x $ and $\tuple y$, 
 $\tuple x \cap \tuple y$ for any tuple listing the variables that occur both in $\tuple x$ and $\tuple y$, and $\tuple x \setminus \tuple y$ for any tuple listing the variables that occur in $\tuple x $ but not in $\tuple y$. For an assignment $s$, we write $s(\tuple x)$ to denote the sequence $\big(s(x_1),\ldots ,s(x_n)\big)$. 

Next we define the lax and strict team semantics of first-order logic. It is worth noting that the disjunction has a non-classical interpretation. The classical (or intuitionistic) disjunction is usually denoted by $\ovee$ in the  team semantics framework. However, as exemplified by Proposition \ref{extension}, the non-classical disjunction of team semantics naturally corresponds to the classical disjunction of ordinary first-order logic.
\begin{definition}[Lax team semantics]
Let $\A$ be a $\tau$-structure and $X$ a team of $\A$. The satisfaction relation $\models_X$ for first-order logic is defined as follows:

\begin{tabbing}
left \= $\A \models_{X} (\psi \land \theta)$ \= $\Leftrightarrow$ \= $\forall s\in X: s(\tuple x) \in R^{\A}$\kill
\> $\A \models_{X} x=y$ \> $\Leftrightarrow$ \> $\forall s\in X: s(x)=s(y)$\\ 
\> $\A \models_{X} x\neq y$ \> $\Leftrightarrow$ \> $\forall s\in X: s(x) \not= s(y)$\\
\> $\A \models_{X} R(\tuple x)$ \> $\Leftrightarrow$ \> $\forall s\in X: s(\tuple x) \in R^{\A}$\\ 
\> $\A \models_{X} \neg R(\tuple x)$ \> $\Leftrightarrow$ \> $\forall s\in X: s(\tuple x) \not\in R^{\A}$\\
\> $\A\models_{X} (\psi \land \theta)$ \> $\Leftrightarrow$ \> $\A \models_{X} \psi \text{ and } \A \models_{X} \theta$\\
\> $\A\models_{X} (\psi \lor \theta)$ \> $\Leftrightarrow$ \> $\A\models_Y \psi \text{ and } \A \models_Z \theta \text{ for some  $Y,Z\sub X$}$ s.t.\ $Y\cup Z = X$\\
\> $\A\models_{X} \forall x\psi$ \> $\Leftrightarrow$ \> $\A\models_{X[A/x]} \psi$\\
\> $\A\models_{X} \exists x\psi$ \> $\Leftrightarrow$ \> $\A\models_{X[F/x]} \psi \text{ holds for some } F\colon X \to \Po(A)\setminus \{\emptyset\}$.
\end{tabbing}
\end{definition}
The so-called \emph{strict} team semantics is obtained from the previous definition by adding the following two requirements.
\begin{enumerate}[(i)]
\item Disjunction: $Y\cap Z= \emptyset$.
\item Existential quantification: $F(s)$ is singleton for all $s\in X$.
\end{enumerate}

\begin{proposition}[\cite{vaananen07}]\label{extension}
Let $\A$ be a $\tau$-structure, $X$ a team of $\A$, and $\varphi$ a formula of $\FO(\tau)$. Then
\[
\A\models_X \varphi \Leftrightarrow \forall s\in X: \A\models_s \varphi,
\]
where $\models_s$ denotes the ordinary satisfaction relation of first-order logic defined via models and assignments as usual, and $\models_X$ denotes the satisfaction relation of either lax or strict team semantics.
\end{proposition}

For a model $\A$ and a \emph{sentence} $\varphi$ (i.e., a formula with no free variables), the satisfaction relation $\models$ is defined as: 
$$\A \models \varphi\textrm{ if }\A \models_{\{\emptyset\}} \varphi,$$
where $\{\emptyset\}$ denotes the singleton team of empty assignment.

Team semantics enables extending first-order logic with various dependency notions. The following dependency atoms were introduced in team semantics setting in \cite{vaananen07,Galliani:2011,gradel13}.

\begin{definition}[Dependency atoms]
Let $\A$ be a model and $X$ a team of $\A$. If $\tuple x, \tuple y$ are variable sequences, then $\dep(\tuple x,\tuple y)$ is a dependence atom with the satisfaction relation:
$$\A \models_X \dep(\tuple x,\tuple y) \textrm{ if for all } s,s'\in X \textrm{ s.t.\ }s(\tuple x) = s'(\tuple x), \textrm{ it holds that } s(\tuple y)=s'(\tuple y).$$

If $\tuple x,\tuple y$ are variable sequences of the same length, then $\tuple x \subseteq \tuple y$ is an inclusion atom with the satisfaction relation:
$$\A \models_X \tuple x \sub \tuple y \textrm{ if for all }s\in X \textrm{ there exists } s'\in X\textrm{ such that } s(\tuple x)=s'(\tuple y).$$ 
If $\tuple x,\tuple y,\tuple z$ are variable sequences, then $\cixyz$ is a conditional independence atom with the satisfaction relation:
\begin{align*}&\A\models_{X} \cixyz \textrm{ if for all } s,s'\in X\textrm{ such that }s(\tuple x)=s'(\tuple x)\textrm{ there exists } s''\in X \\&\textrm{ such that }s''(\tuple x)=s(\tuple x)\textrm{, }s''(\tuple y)=s(\tuple y), \textrm{ and }s''(\tuple z)=s'( \tuple z). 
\end{align*}
\end{definition}

Note that in the previous definition it is allowed that some or all of the vectors of variables have length $0$. For example, $\A\models_X\dep(\tuple x)$ holds iff $\forall s\in X: s(\tuple x)=\tuple c$ holds for some fixed tuple $\tuple c$, and $\A\models_X\cixyz$ holds always if either  of the vectors $\tuple y$ or $\tuple z$ is of length $0$.

We write $\FO$ for first-order logic, and given a set of atoms $\mathcal{C}$,
we write $\FO(\mathcal{C})$ (omitting the set parentheses of $\mathcal{C}$) for the logic obtained by adding the atoms of $\mathcal{C}$ to $\FO$. For instance, $\FO(\dep(\cdot))$ denotes then dependence logic.

Often in literature dependence atoms are defined such that $\tuple y$ is a single variable, i.e., the widely used form is $\dep(\tuple x,y)$. The definition above yields the strongest form of \emph{functional dependence}. Moreover the atom $\dep(\tuple x,\tuple y)$ can be equivalently rewritten as a conjunction of dependence atoms of type $\dep(\tuple x,y)$.

\section{Multiteam Semantics}\label{Mult}

In this section we generalise team semantics with the concept of multisets.  Multisets and multiteam semantics can be used, e.g., in applications to database theory to model reasoning with databases with duplicates. In practice, for multitude of reasons, the existence of duplicates in databases is very common.  Again as previously noted, we restrict attention to finite sets and finite multisets. In the following definition, occurrences of ``zero multiplicities'' are allowed for notational convenience.
\begin{definition}[Multiset]
A \emph{multiset} is a pair $(A,m)$ where $A$ is a set and $m:A \to \N$ is a (multiplicity) function. 
The function $m$ determines the \emph{multiplicities} of the elements in the multiset $(A,m)$. A multiset $(X,m)$ is a \emph{multiteam} if the underlining set $X$ is a team. The domain (or the codomain) of the multiteam $(X,m)$ is the domain (codomain) of the team $X$. 
\end{definition}
For each multiset $(A,m)$, we define the \emph{canonical set representative} $\cset{(A,m)}$ of $(A,m)$ as follows: 
\[
\cset{(A,m)} := \{ (a,i) \mid a\in A, 0<i\leq m(a)\}.
\]

We say that $(A,m)$ is finite whenever $\cset{(A,m)}$ is finite. We say that a multiset $(A,m)$ is a submultiset of a multiset $(B,n)$, $(A,m) \subseteq (B,n)$, if and only if $\cset{(A,m)} \subseteq \cset{(B,n)}$. Furthermore, we define that $(A,m) = (B,n)$ if and only if both $(A,m) \subseteq (B,n)$ and $(B,n) \subseteq (A,m)$ hold.

The \emph{disjoint union} $(A,m) \uplus (B,n)$ of $(A,m)$ and $(B,n)$ is the multiset $(C,k)$, where $C:= A\cup B$ and $k:C\to\N$ is the function defined as follows:
\[
k(s):=
\begin{cases}
m(s)+n(s) & \text{if $s\in A$ and $s\in B$},\\
m(s) & \text{if $s\in A$ and $s\not\in B$},\\
n(s) & \text{if $s\not\in A$ and $s\in B$}.
\end{cases}
\]

We write $\lvert(A,m)\rvert$ to denote the size of the multiset $(A,m)$, i.e., $\lvert(A,m)\rvert := \sum_{a\in A} m(a)$.
The set of non-empty submultisets of a multiset $(A,m)$ is the set 
$$\Po^+\big((A,m)\big):=\{ (C,l) \mid (C,l) \subseteq (A,m) \text{ s.t. } l(c)\geq 1 \text{ for each }c\in C\}\setminus\{(\emptyset,\emptyset)\}.$$
Let $(X,m)$ be a multiteam, $(A,n)$ a finite multiset, and $F\colon \cset{(X,m)}\to \Po^+\big( (A,n) \big)$ a function. We denote by $(X,m)[(A,n)/x]$ the modified multiteam defined as 
\[
\biguplus_{s\in X}\biguplus_{a\in A} \{\big(s(a/x), m(s) \cdot n(a)\big)\}.
\]
By $X[F/x]$ we denote the multiteam defined as
\[
\biguplus_{s\in X}\biguplus_{1\leq i \leq m(s)} \{\big(s(b/x), l(b)\big) \mid (B,l) = F\big((s,i)\big), b\in B\}.
\]

A $\tau$-multistructure is a tuple $\A = \big((A,m), (R^\A_i)_{R_i\in\tau}\big)$ where $(A,m)$ is a multiset and, for each $R_i\in\tau$, $R^\A_i$ is an $\ar(R_i)$-ary relation over the set $\{a\in A\mid m(a) \geq 1 \}$. A multiteam $(X,m)$ over $\A$ is a multiteam with codomain $A$. 

Next we define multiteam semantics for first-order logic.

\begin{definition}[Multiteam semantics]\label{laxmulti}
Let $\A$ be a $\tau$-multistructure, $(A,n)$ the domain of $\A$, and $(X,m)$ a multiteam over $\A$. The satisfaction relation $\models_{(X,m)}$ is defined as follows:
\begin{tabbing}
left \= $\A \models_{(X,m)} (\psi \land \theta)$ \= $\Leftrightarrow$ \= $\forall s\in X: \text{ if } m(s)\geq 1 \text{ then }  s(\tuple x) \not\in R^{\A}$\kill
\> $\A \models_{(X,m)} x=y$ \> $\Leftrightarrow$ \> $\forall s\in X: \text{ if } m(s)\geq 1 \text{ then }  s(x)=s(y)$\\ 
\> $\A \models_{(X,m)} x\neq y$ \> $\Leftrightarrow$ \> $\forall s\in X: \text{ if } m(s)\geq 1 \text{ then }  s(x)\not=s(y)$\\
\> $\A \models_{(X,m)} R(\tuple x)$ \> $\Leftrightarrow$ \> $\forall s\in X: \text{ if } m(s)\geq 1 \text{ then }  s(\tuple x) \in R^{\A}$\\ 
\> $\A \models_{(X,m)} \neg R(\tuple x)$ \> $\Leftrightarrow$ \> $\forall s\in X: \text{ if } m(s)\geq 1 \text{ then }  s(\tuple x) \not\in R^{\A}$\\
\> $\A\models_{(X,m)} (\psi \land \theta)$ \> $\Leftrightarrow$ \> $\A \models_{(X,m)} \psi \text{ and } \A \models_{(X,m)} \theta$\\
\> $\A\models_{(X,m)} (\psi \lor \theta)$ \> $\Leftrightarrow $\> $\A\models_{(Y,k)} \psi  \text{ and } \A \models_{(Z,l)} \theta \text{ for some multisets}$\\
\>\>\> $\text{$(Y,k),(Z,l)\subseteq (X,m)$ s.t. $(X,m) \subseteq (Y,k)\uplus(Z,l)$.}$\\
\> $\A\models_{(X,m)} \forall x\psi$ \> $\Leftrightarrow$ \> $\A\models_{(X,m)[(A,n)/x]} \psi$\\
\> $\A\models_{(X,m)} \exists x\psi$ \> $\Leftrightarrow$ \> $\A\models_{(X,m)[F/x]} \psi \text{ holds for some function}$\\
\>\>\> $F\colon \cset{(X,m)} \to \Po^+\big((A,n)\big)$.
\end{tabbing}
\end{definition}
The so-called \emph{strict multiteam semantics} is obtained from the previous definition by adding the following  two requirements.  
\begin{enumerate}[(i)]
\item Disjunction: $(Y,n)\uplus (Z,k) = (X,m)$.
\item Existential quantification: for all $s\in X$ and $0<i\leq m(s)$, $F\big( (s,i) \big)=(B,n)$ for some singleton $B=\{b\}$ and $n(b)=1$.
\end{enumerate}

 This alternative semantics is discussed in Section \ref{strictsect}. Otherwise in the paper we restrict attention to the multiteam semantics given in Definition \ref{laxmulti}, sometimes referred to as \emph{lax multiteam semantics}. The following proposition shows that  multiteam semantics and  team semantics for first-order logic coincide when the multisets in multistructures are essentially sets. The proof of the proposition is self evident.

\begin{proposition}\label{prop:sets}
Let $\A$ be a multistructure with domain $(A,n)$, and $(X,m)$ a multiteam over $\A$ such that $n(a)=m(s)=1$ for all $a\in A$ and $s\in X$. Define $\mathfrak B:=(A, (R^\A)_{R\in\tau})$. Then for every $\varphi\in\FO$ it holds that
$$
\A\models_{(X,m)} \varphi \text{ if and only if } \mathfrak B\models_{X}\varphi.
$$
\end{proposition}

Next we generalise inclusion and conditional independence atoms to multiteams by introducing their probabilistic versions. For a multiteam $(X,m)$ of codomain $A$,
a tuple of variables $\tuple x$ from $\Dom(X)$, and $\tuple a \in A^{|\tuple x|}$, we denote by $(X,m)_{\tuple x= \tuple a}$ the multiteam $(X,n)$ where $n$ agrees with $m$ on all assignments $s\in X$ with $s(\tuple x)=\tuple a$, and otherwise $n$ maps $s$ to $0$.


\begin{definition}\label{probsem}
Let $\A$ be a multistructure with domain $(A,n)$, and $(X,m)$ a multiteam over $\A$. If $\tuple x,\tuple y$ are variable sequences of the same length, then $\tuple x \leq \tuple y$ is a \emph{probabilistic inclusion atom} with the following semantics:
$$\A \models_{(X,m)} \tuple x \leq \tuple y \textrm{ if } |(X,m)_{\tuple x=\tuple s(\tuple x)}| \leq |(X,m)_{\tuple y=s(\tuple x)}| \textrm{ for all } s:\Var(\tuple x) \to A.$$
If $\tuple x,\tuple y,\tuple z$ are variable sequences, then $\pcixyz$ is a \emph{probabilistic conditional independence atom} with the satisfaction relation defined as
\begin{align}
	\A\models_{(X,m)} \pcixyz \label{def1}
\end{align}
if for all $s\colon\Var(\tuple x\tuple y\tuple z) \to A$ it holds that
\begin{align}
|(X,m)_{\tuple x\tuple y= s(\tuple x\tuple y)}| \cdot  |(X,m)_{\tuple x\tuple z=s(\tuple x\tuple z)}|=|(X,m)_{\tuple x\tuple y\tuple z=s(\tuple x\tuple y\tuple z)}|\cdot |(X,m)_{\tuple x=s(\tuple x)}|.\nonumber
\end{align}
\end{definition}

We call atoms of the form $\pci{\emptyset}{\tuple x}{\tuple y}$ \emph{probabilistic marginal independence atoms}, written as the shorthand $\pmi{\tuple x}{\tuple y}$. Note that we obtain the following satisfaction relation for $\pmi{\tuple x}{\tuple y}$:
\begin{align}
\A\models_{(X,m)} \pmixy \textrm{ if for all } &s:\Var(\tuple x\tuple y) \to A, \label{def2}\\
\frac{|(X,m)_{\tuple x =s(\tuple x)}| \cdot  |(X,m)_{\tuple y=s(\tuple y)}|}{|(X,m)|}&=|(X,m)_{\tuple x\tuple y=s(\tuple x\tuple y)}|.\nonumber
\end{align}

The study of database dependencies is very interesting also in the practical point of view as many interesting properties of datasets can be revealed.
Further the investigation of conditional independence can yield methods to be used to decompose datasets for speeding up different processing tasks on the data.

Multiteams $(X,m)$ induce a natural probability distribution $p$ over the assignments of $X$. Namely, we define $p\colon X \to [0,1]$ such that 
$$p(s) = \frac{m(s)}{\sum_{s\in X} m(s)}.$$ 

The probability that a tuple of (random) variables $\tuple x$ takes value $\tuple a$, written $\Pr(\tuple x = \tuple a)$, is then 
$$\sum_{\substack{s\in X,\\s(\tuple x) = \tuple a}} p(s).$$ 
It is now easy to see that $\A\models_{(X,m)} \pcixyz$ iff for all $\tuple a\tuple b\tuple c$,
$$\Pr(\tuple y=\tuple b, \tuple z = \tuple c | \tuple x= \tuple a) = \Pr(\tuple y=\tuple b | \tuple x= \tuple a)\Pr (\tuple z = \tuple c |  \tuple x= \tuple a),$$
that is, the probability of $\tuple y=\tuple b$ is independent of the probability of $\tuple z = \tuple c$, given  $\tuple x=\tuple a$. Analogously, a probabilistic inclusion atom $\tuple x \le \tuple y$ indicates that $\Pr(\tuple x = \tuple a) \leq \Pr(\tuple y =\tuple a)$ for all values $\tuple a$, and a probabilistic independence atom of the form $\pmi{\tuple x}{\tuple x}$ that $\Pr(\tuple x = \tuple a)=1$ for some value $\tuple a$. Note that such atoms have been studied in the literature under the name of \emph{constancy atoms} \cite{Galliani:2011}.

One can also study the usual dependency notions of database theory in the multiteam semantics setting.
\begin{definition}
Let $\A$ be a multistructure, $(X,m)$ a multiteam over $\A$, and $\varphi$ of the form $\dep(\tuple x, \tuple y)$, $\tuple x \sub \tuple y$, or $\cixyz$. Then the satisfaction relation $\models_{(X,m)}$ is defined as follows:
$$\A \models_{(X,m)} \varphi \,\textrm{ iff }\, \A \models_{X^+} \varphi,$$
where $X^+$ is the team $\{s\in X \mid m(s) \geq 1\}$. 
\end{definition}

First we notice that the known translation of dependence atoms to independence atoms (see Gr\"adel et~al.\ \cite{gradel13}) works also in the probabilistic case.
\begin{proposition}\label{dep}
Let $\A$ be a multistructure, $(X,m)$ a multiteam over $\A$, and $\tuple x,\tuple y$ tuples of variables. Then $\A \models _{(X,m)}\pci{\tuple x}{\tuple y}{\tuple y}$ \,iff\, $\A \models_{(X,m)} \dep(\tuple x,\tuple y)$.
\end{proposition}
\begin{proof}
From the truth definition we obtain that
\begin{align}
\A\models_{(X,m)} \pci{\tuple x}{\tuple y}{\tuple y}  \,\Leftrightarrow\,\,  &\textrm{for all } s:\Var(\tuple x\tuple y) \to A\textrm{ with }(X,m)_{\tuple x\tuple y =s(\tuple x\tuple y)} \neq \emptyset , \label{pcidep}\\   
&|(X,m)_{\tuple x\tuple y =s(\tuple x\tuple y)}| = |(X,m)_{\tuple x =s(\tuple x)}|. \nonumber
\end{align}
The result then follows since $ \A \models_{(X,m)} \dep(\tuple x,\tuple y)$ iff the right-hand side of \eqref{pcidep} holds. \qed 
\end{proof}
Note that the restriction of Proposition \ref{dep} to marginal independence states that
$$
\A \models _{(X,m)} \pmi{\tuple x}{\tuple x} \quad\Leftrightarrow\quad  \A \models_{(X,m)} \dep(\tuple x).
$$

It is left open whether one can define inclusion or conditional independence atoms in $\FO(\bbot_{\rm c},\leq)$. However, over constant multiplicity functions conditional independence atoms $\varphi$ coincide with their probabilistic counterparts whenever $\Var(\varphi) =\Dom(X)$.
In the following, we denote by $\leftsup{\tuple x}{A}$ the team of all assignments $\Var(\tuple x) \to A$.
\begin{lemma}\label{rules}
Let $\A$ be a multistructure and $(X,m)$ a multiteam over $\A$. Then
\begin{enumerate}[(i)]
\item $\A \models_{(X,m)} \pci{\tuple x}{\tuple y}{\tuple z} \quad\Leftrightarrow\quad \A \models_{(X,m)} \pcib{\tuple x}{\tuple y\setminus \tuple x}{\tuple z\setminus \tuple x}$, 
\item $\A \models_{(X,m)} \pci{\tuple x}{\tuple y}{\tuple z} \quad\Leftrightarrow\quad \A \models_{(X,m)} \pcib{\tuple x}{\tuple y\setminus \tuple z}{\tuple z\setminus \tuple y} \wedge \pcib{\tuple x}{\tuple y\cap\tuple z}{\tuple y\cap\tuple z}$.
\end{enumerate}
\end{lemma}
\begin{proof}
\textbf{Case (i)}. The truth definition in \eqref{def1} is symmetric, and hence it suffices to show that $\A \models_{(X,m)} \pci{\tuple x}{\tuple y x}{\tuple z} \Leftrightarrow  \A \models_{(X,m)} \pci{\tuple x}{\tuple y}{\tuple z}$ whenever $x$ is listed in $\tuple x$. This follows since $\leftsup{\tuple x\tuple y x \tuple z}{A}= \leftsup{\tuple x\tuple y \tuple z}{A}$, and the equation in \eqref{def1} remains the same after removing $x$.

\textbf{Case (ii)}. Let us first show that $\A \models_{(X,m)} \pci{\tuple x}{\tuple y}{\tuple z}$ implies $\A \models_{(X,m)} \pcib{\tuple x}{\tuple y \cap\tuple z}{\tuple y \cap \tuple z}$. For this, it remains to show that $\A \models_{(X,m)} \pci{\tuple x}{\tuple y u}{\tuple z}$ implies $\A \models_{(X,m)} \pci{\tuple x}{\tuple y}{\tuple z}$, for $u$ not listed in $\tuple x \tuple y\tuple z$. This follows since for all $s\in \leftsup{\tuple x\tuple y \tuple z}{A}$, 
\begin{align*}
   |(X,m)_{\tuple x\tuple z=s(\tuple x\tuple z)}| &\cdot |(X,m)_{\tuple x\tuple y= s(\tuple x\tuple y)}|\\
&=  |(X,m)_{\tuple x\tuple z=s(\tuple x\tuple z)}| \cdot \Sigma_{a\in A} |(X,m)_{\tuple x\tuple yu= s(\tuple x\tuple y)a}|\\
&=\Sigma_{a\in A} (|(X,m)_{\tuple x\tuple z=s(\tuple x\tuple z)}|\cdot |(X,m)_{\tuple x\tuple yu= s(\tuple x\tuple y)a}|)\\
& = \Sigma_{a\in A} ( |(X,m)_{\tuple x = s(\tuple x)}|\cdot|(X,m)_{\tuple x\tuple y\tuple z u=s(\tuple x\tuple y\tuple z)a}|  )\\
&= |(X,m)_{\tuple x = s(\tuple x)}|\cdot  \Sigma_{a\in A}|(X,m)_{\tuple x\tuple y\tuple z u=s(\tuple x\tuple y\tuple z)a}| \\
&= |(X,m)_{\tuple x = s(\tuple x)}| \cdot |(X,m)_{\tuple x\tuple y\tuple z=s(\tuple x\tuple y\tuple z)}|,
\end{align*}
where in the third equation we apply the assumption that $\A \models_{(X,m)} \pci{\tuple x}{\tuple y a}{\tuple z}$. 

For the claim it now suffices to show that $\A \models_{(X,m)} \pci{\tuple x}{\tuple y}{\tuple z} \Leftrightarrow \A \models_{(X,m)} \pcib{\tuple x}{\tuple y\setminus \tuple z}{\tuple z\setminus \tuple y}$ whenever $\A \models_{(X,m)} \pcib{\tuple x}{\tuple y \cap\tuple z}{\tuple y \cap \tuple z}$. This follows directly from the truth definition since by \eqref{pcidep} for all $s\in \leftsup{\tuple x\tuple v}{A}$ with $(X,m)_{\tuple x\tuple v =s(\tuple x\tuple v)} \neq \emptyset$:
$$|(X,m)_{\tuple x\tuple v =s(\tuple x\tuple v)}| = |(X,m)_{\tuple x =s(\tuple x)}|,$$
for $\tuple v:= \tuple x \cap \tuple y$. \qed
\end{proof}

If $\tuple x,\tuple y,\tuple z$ are pairwise disjoint, then $\pcixyz$ corresponds to the generalised embedded multivalued dependency $\tuple x\multimap\to \tuple y\mid\tuple z$ that is defined over extended relational data models (i.e., relational data models equipped with a multiplicity function) using semantics that coincide with that of Definition \ref{probsem} \cite{wong97,wbw00}. It was shown by Wong \cite{wong97} that the generalised multivalued dependency $\tuple x\multimap\to \tuple y$ holds in an extended relational data model if and only if the underlying relational model satisfies the multivalued dependency $\tuple x\twoheadrightarrow \tuple y$. This is stated in the following theorem reformulated into the framework of this article.

\begin{theorem}[\cite{wong97}]\label{wongthm} Let $\A$ be a multistructure, $X$ a team over $\A$, and $\pcixyz$ a probabilistic conditional independence atom such that $\Var(\pcixyz) =\Dom(X)$ and $\tuple x ,\tuple y,\tuple z$ are pairwise disjoint. Let $1$ denote the constant function that maps all assignments of $X$ to $1$.
Then $\A \models_{(X,1)} \pcixyz$ \,iff\, $\A \models_{(X,1)} \cixyz$.
\end{theorem}
Using Lemma \ref{rules}, the restriction that $\tuple x,\tuple y,\tuple z$ are disjoint can be now removed.
\begin{proposition}\label{equivalence_of_probabilistic_and_possibilistic_independence}
Let $\A$ be a multistructure, $X$ a team over $\A$, and $\pcixyz$ a probabilistic conditional independence atom such that $\Var(\pcixyz) =\Dom(X)$. Then $\A \models_{(X,1)} \pcixyz$ \,iff\, $\A \models_{(X,1)} \cixyz$.
\end{proposition}
 \begin{proof}
First note that by Proposition \ref{dep} and Lemma \ref{rules}, $\pcixyz$ is equivalent in multiteam semantics to $\pcib{\tuple x}{\tuple y\setminus \tuple x\tuple z}{\tuple z \setminus \tuple x \tuple y} \wedge \dep(\tuple x,\tuple y\cap\tuple z)$. Moreover, it is known that in team semantics $\cixyz$ is equivalent to $\ci{\tuple x}{\tuple y\setminus \tuple x\tuple z}{\tuple z \setminus \tuple x \tuple y} \wedge \dep(\tuple x,\tuple y\cap\tuple z)$ \cite{gradel13}. Hence, the claim follows by Theorem \ref{wongthm}.\qed
\end{proof}

Note that $\pcixyz$ implies $\cixyz$ also over arbitrary multiplicity functions since non-emptiness of $(X,m)_{\tuple x\tuple y=s(\tuple x\tuple y)}$ and $(X,m)_{\tuple x\tuple z=s(\tuple x\tuple z)}$ implies non-emptiness of $(X,m)_{\tuple x\tuple y\tuple z=s(\tuple x\tuple y\tuple z)}$ by the truth definition in \eqref{def1}. The converse however does not hold; the multiteam $(Y,m)$ depicted in Fig. \ref{pic1} satisfies $\mi{x}{y}$ but violates $\pmi{x}{y}$.

\subsubsection{A diversion: implication problems.}
Results similar in spirit to Proposition~\ref{equivalence_of_probabilistic_and_possibilistic_independence} have been studied in connection to \emph{implication problems} which is a central notion in causal reasoning and database dependency theory. The finite implication problem of independence atoms $\cixyz$ is defined as follows. Given a finite collection $\Sigma\cup \{\varphi\}$ of independence atoms, determine whether for all finite $\A,X$: 
$$\A\models_X \Sigma \Rightarrow \A\models_X\varphi.$$ 

If the above holds, we write $\Sigma\models\varphi$. The implication problem of other types of dependencies is defined analogously. Furthermore, the problem for the atoms $\pcixyz$ can be defined similarly by replacing teams by multiteams. The implication problems of embedded multivalued dependencies (i.e., the atoms $\cixyz$)  and $\pcixyz$ have been extensively studied, e.g., for both atoms the problem is not finitely axiomatisable and for the former the problem is known to be undecidable. On the other hand, there are interesting restricted cases where the implication problems are finitely axiomatisable and equivalent, i.e., for all inputs $\Sigma\cup \{\varphi\}$, $\Sigma\models\varphi$ iff $\Sigma^p\models\varphi^p$, where $\Sigma^p$ and $\phi^p$ are defined by replacing $\cixyz$ by $\pcixyz$. This holds, for example, for marginal independence atoms and for the so-called \emph{saturated} atoms $\varphi$ of the form $\pcixyz$ (or equivalently for $\varphi=\cixyz$) which, as in Theorem \ref{wongthm}, satisfy $\Var(\pcixyz) =\Dom(X)$
 (see the survey by Wong et~al.\ \cite{wbw00}).
Relationships between fragments of conditional independence statements and  embedded multivalued dependencies have recently been  studied, e.g.,  in \cite{DBLP:conf/aaai/Link13,Link13,Link15A,Niepert201329}. Is is also worth noting that the passage from set to  multisets  has interesting consequences also for the study of implication problems of  database dependencies. For example,  while key constraints can be expressed by functional dependencies under team semantics, this is no longer true under multiteam semantics  \cite{DBLP:journals/ipl/KohlerL10}.

Conditional independence is an important notion for expressing structural aspects of probability distributions. \emph{Context specific independence} is a variant of $\pcixyz$ expressing independence in a context where the values of some variables of $\tuple x$ are restricted to range over a subset of all possible values \cite{Boutilier:1996:CIB:2074284.2074298,DBLP:journals/datamine/PensarNKC15}. The next simple example shows how disjunction can be used to express context specific independence statements in $\FO( \bbot_{\rm c})$. The example shows that combining $\pcixyz$ with the logical connectives and quantifiers available in $\FO( \bbot_{\rm c})$ provide us with powerful means to define interesting generalisations of conditional independence. The definability of context specific independence using disjunction has been pointed out in  \cite{CKV2015}.

\begin{example} Let $A=\{0,1\}$ and $X$ be a multiteam of $A$ with domain $\Dom(X)=\{x_0,x_1, \ldots ,x_n\}$. Now in $X$ the variable $x_0$ is said to be contextually independent of $x_2$ given $x_1=0$, denoted by
 \begin{equation}\label{csi}
  x_0 \perp x_2 \mid x_1=0,
  \end{equation}    
 if for all $s\colon \{x_0,x_1,x_2\}\rightarrow A$ such that $s(x_1)=0$ it holds that 
\begin{align*} 
|(X,m)_{x_0x_1= s(x_0x_1)}| \cdot  &|(X,m)_{x_1x_2=s(x_1x_2)}|\\
&=|(X,m)_{x_0x_1x_2=s(x_0x_1x_2)}|\cdot|(X,m)_{x_1=s(x_1)}|.	
\end{align*}

It is now straightforward to check that \eqref{csi} can be equivalently expressed by the $\FO( \bbot_{\rm c})$-formula
$(x_1\neq c) \vee \big(x_1=c \wedge (x_0 \bbot_{x_1} x_2)\big),$ where $c$ is a constant symbol interpreted as $0$. 
\end{example}

\subsection{Probabilistic Notions in Multiteam Semantics}
In this section we investigate some properties of the probabilistic logics we have defined so far.

The set of \emph{free variables} of a formula $\varphi\in\FO(\mathcal{C})$, denoted by $\Fr(\varphi)$, is defined in the obvious manner as in first-order logic. In particular, we define
\begin{align*}
\Fr(\vec x \sub \vec y) &:= \Fr(\vec x \leq \vec y) := \Fr\big(\dep(\vec x,\vec y)\big) := \{\vec x, \vec y\}\\
\Fr\big( \vec y\bbot_{\vec x} \vec z  \big) &:= \Fr\big( \vec y\bot_{\vec x} \vec z  \big) :=  \{\vec x, \vec y, \vec z \}.
\end{align*} 
For $V \sub \Dom(X)$, we define $(X,m) \upharpoonright V:=(X\upharpoonright V, \,n)$
where
$$n(s):= \sum_{\substack{s'\in X,\\ s'\upharpoonright V = s}} m(s').$$
The following locality principle holds by easy structural induction.
\begin{proposition}[Locality]\label{locality}
 Let $\A$ be a multistructure, $(X,m)$ a multiteam, and $V$ a set of variables such that $\Fr(\varphi) \sub V\sub \Dom(X)$. Then for all $\varphi\in\FO(\leq,\linebreak\bbot_{\rm c},\dep(\cdot),\sub,\bot_{\rm c})$ it holds that
$\A\models_{(X,m)}\varphi$ \,iff\, $\A\models_{(X,m)\upharpoonright V}\varphi$.
\end{proposition}
The notion of flatness is generalised to the multiteam setting as follows.
\begin{definition}[Weak flatness]
We say that a formula $\varphi$ is \emph{weakly flat} if for all multistructures $\A$ and for all multiteams $(X,m)$ it holds that
$$\A \models_{(X,m)} \varphi \quad\Leftrightarrow\quad \A\models_{(X,n)} \varphi,$$
where $n$ agrees with $m$ on all $s$ with $m(s) =0$, and otherwise maps all $s$ to $1$. The multiteam $(X,n)$ is then called the \emph{weak flattening} of $(X,m)$. A logic is called \emph{weakly flat} if every formula of this logic is weakly flat.
\end{definition}

 Dependence, conditional independence, and inclusion atoms are insensitive to multiplicities, and using structural induction one can prove the following proposition.
\begin{proposition}
$\FO(\dep(\cdot),\sub,\bot_{\rm c})$ is weakly flat.
\end{proposition}

On the other hand, probabilistic dependencies do not satisfy weak flatness as shown in the next example. 

\begin{example}\label{ex:pic1}
For instance $(Y,m)$, illustrated in Fig.~\ref{pic1}, does not satisfy $\pmixy$ but its weak flattening $(Y,n)$ does. 	
\end{example}

Analogously, the probabilistic inclusion atom is not weakly flat, and therefore neither of these atoms can be expressed in $\FO(\dep(\cdot),\sub,\bot_{\rm c})$.

\begin{figure}[ht]
\begin{center}
 \scalebox{1.2}[1.1]{
\begin{tabular}{cC{3mm}C{3mm}C{9mm}}
\multicolumn{4}{c}{$(Y,m)$}\\\toprule
&$x$& $y$ & $m(s_i)$\\\midrule
$s_0$&$0$ & $0 $ & $2$ \\
$s_1$&$0$ &$1 $ & $1$\\
$s_2$&$1$ &$ 0$ & $1$\\
$s_3$&$1$ &$ 1$ & $1$\\\bottomrule

\end{tabular}

\qquad

\begin{tabular}{cC{3mm}C{3mm}C{9mm}}
\multicolumn{4}{c}{$(Y,n)$}\\\toprule
&$x$& $y$ & $n(s_i)$\\\midrule
$s_0$&$0$ & $0 $ & $1$ \\
$s_1$&$0$ &$1 $ & $1$\\
$s_2$&$1$ &$ 0$ & $1$\\
$s_3$&$1$ &$ 1$ & $1$\\\bottomrule

\end{tabular}
}
\caption{Assignments for multiteams in Example~\ref{ex:pic1}.\label{pic1}}
\end{center}
\end{figure}

A formula $\varphi$ is called \emph{union closed} (in the multiteam setting) if for all multistructures $\A$ and all multiteams $(X,m),(Y,n)$: if $\A\models_{(X,m)} \varphi$ and $\A\models_{(Y,n)} \varphi$, then $\A\models_{(Z,h)}\varphi$, where $(Z,h)=(X,m)\uplus (Y,n)$.  
A logic is called union closed if all its formulae are union closed. It is easy to show by induction on the structure of formulae that probabilistic inclusion logic satisfies union closure.
\begin{proposition}\label{multiuc}
$\FO(\leq,\sub)$ is union closed.
\end{proposition}

\subsection{Probabilistic Notions in Team Semantics}
In this section we examine probabilistic independence and inclusion logic in the (set) team semantics setting. Note that all the models considered in this section are usual first-order structures.

Satisfaction of probabilistic atoms in team semantics setting  is defined by adding a constant multiplicity function. 
\begin{definition} 
Let $\A$ be a model, $X$ a team over $\A$, and $\varphi$ be a probabilistic atom of the form $\pcixyz$ or $\tuple x \leq \tuple y$. Then the satisfaction relation $\models_X$ is defined as follows:
$$\A \models_X \varphi \textrm{ iff }\A \models_{(X,1)} \varphi,$$
where $1$ is the constant function that maps all assignments of $X$ to $1$. 
\end{definition}

The next theorem shows that, since probabilistic inclusion and independence atoms are expressible (in the team semantics setting) in $\FO(\bot_{\rm c})$ relative to teams of fixed domain, their addition does not increase the expressive power of $\FO(\bot_{\rm c})$.
\begin{theorem}\label{setthm}
Let $\varphi \in \FO(\leq,\bbot_{\rm c},\dep(\cdot),\sub,\bot_{\rm c})$ be a sentence. Then there exists a sentence $\varphi'\in \FO(\bot_{\rm c})$ such that for all models $\A$ it holds that $\A \models \varphi$ iff $\A \models \varphi'$.
\end{theorem}
\begin{proof}
First note that inclusion and dependence atoms can be expressed in $\FO(\bot_{\rm c})$ \cite{Galliani:2011,gradel13}. Also it is easy to see that one can construct existential second-order logic sentences that capture probabilistic inclusion and conditional independence atoms over teams of fixed domain. Namely, for all $\varphi$ of the form $\pcixyz$ or $\tuple x \leq \tuple y$ and all $ V \supseteq\Fr(\varphi)$, there exists an $\ESO$ sentence $\varphi^*(R)$, where $R$ is a $k$-ary relation symbol for $k= |\Var(\varphi)|$, such that for all $\A$ and $X$ with $\Dom(X) =V$,
$$\A \models_X \varphi \Leftrightarrow (\A,\textrm{Rel}(X))\models \varphi^*(R),$$
where $\text{Rel}(X)= \{(s(x_1), \ldots ,s(x_k))\mid s\in X\}$. All $\ESO$-definable properties of teams translate into $\FO(\bot_{\rm c})$ \cite{Galliani:2011}, and hence the formula $\varphi'$ can be constructed from $\varphi$ by replacing each probabilistic atom with a correct $\FO(\bot_{\rm c})$-translation. \qed
\end{proof}

Note that probabilistic inclusion atoms are not closed under (set) unions in team semantics, and hence they cannot be expressed in $\FO(\sub)$ as shown in the following example. 

\begin{example}\label{ex:propinc_not_union_closed}
Let $\A$ be a first-order structure with domain $\{0,1,2\}$, and $s:=\{(x,0), (y,1), (z,0) \}$, $s':=\{(x,1), (y,0), (z,1) \}$, and $s'':=\{(x,0), (y,1), (z,2) \}$ be assignments. Define $X:=\{s,s'\}$ and $Y:=\{s',s''\}$. Now $\A\models_X x\leq y$, $\A\models_Y x\leq y$, but $\A\not\models_{X\cup Y} x\leq y$. 
\end{example}

\subsection{Strict multiteam semantics}\label{strictsect}
We briefly consider properties of related logics under strict multiteam semantics.
\begin{proposition}
Over strict multiteam semantics $\FO(\dep(\cdot))$ is weakly flat.
\end{proposition}

The logics $\FO(\bot_{\rm c})$ and $\FO(\sub)$ are not weakly flat under strict multiteam semantics as shown in the next example.

\begin{example}\label{ex:pic0}
 For instance $(X,m)$, illustrated in Fig.~\ref{pic0}, satisfies $ (x\sub z)\vee (y\sub z)$ in strict semantics but its weak flattening $(X,n)$ does not.	
\end{example}

\begin{figure}[ht]
\begin{center}
 \scalebox{1.05}[1.05]{
\begin{tabular}{cC{3mm}C{3mm}C{3mm}C{9mm}}
\multicolumn{5}{c}{$(X,m)$}\\\toprule
&$x$& $y$ & $z$ & $m(s_i)$\\\midrule
$s_1$&$0$ & $0 $ & $1$ & $2$ \\
$s_2$&$1$ & $2 $ & $0$ & $1$ \\
$s_3$&$2$ & $1 $ & $0$ & $1$\\\bottomrule
\end{tabular}

\qquad

\begin{tabular}{cC{3mm}C{3mm}C{3mm}C{8mm}}
\multicolumn{5}{c}{$(X,n)$}\\\toprule
&$x$& $y$ & $z$ & $n(s_i)$\\\midrule
$s_1$&$0$ & $0 $ & $1$ & $1$ \\
$s_2$&$1$ & $2 $ & $0$ & $1$ \\
$s_3$&$2$ & $1 $ & $0$ & $1$\\\bottomrule
\end{tabular}
}
\caption{Assignments for teams in Example~\ref{ex:pic0}.\label{pic0}}
\end{center}
\end{figure}

Similarly, one can show that $\FO(\leq,\sub)$ is not union closed under strict multiteam semantics.
Moreover one can show that Propositions \ref{prop:sets} and \ref{locality} hold also under strict multiteam semantics.

\section{Approximate Operators}\label{Appr}
Now we will turn to define an existential and a universal \emph{approximate} operator which allows one to state truth of formulas not with respect to the \emph{full} team but with respect to a \emph{ratio} of the team. 
The main motivator for this approach is the important application in database theory to be able to model the truth of properties in databases that may contain some faulty data. 
Moreover, in practice, duplicates occur frequently in databases for a multitude of reasons. Thus the study of database dependencies, such as inclusion dependencies and foreign key constraints, in combination with approximate operators is an important topic as it explains inherent properties of a given dataset. In this section we consider multiteam semantics.

\begin{definition}
Let $\A$ be a multistructure, and $(X,m)$ a multiteam over $\A$, and $p\in[0,1]$ a rational number.
\begin{align*}
 \A\models_{(X,m)}\F p\varphi &\Leftrightarrow \exists (Y,n)\subseteq (X,m), |(Y,n)|\ge p\cdot|(X,m)|: \A\models_{(Y,n)}\varphi,\\
  \A\models_{(X,m)}\G p\varphi &\Leftrightarrow \forall (Y,n)\subseteq (X,m), |(Y,n)|\ge p\cdot|(X,m)|: \A\models_{(Y,n)}\varphi
\end{align*}
\end{definition}

The previous definition generalises the notion of approximate dependence atoms $\dep_p(\cdot)$, introduced by V{\"a}{\"a}n{\"a}nen \cite{v14}, in the following sense: $\dep_{1-p}(\tuple x,y)$ is equivalent to the formula $\F p\dep(\tuple x,y)$.

In the following we observe that distributivity does not hold in general with respect to $\F p$.
\begin{figure}
\begin{center}
 \scalebox{1.05}[1.05]{
\begin{tabular}{cC{3mm}C{3mm}C{3mm}C{9mm}}
\multicolumn{5}{c}{$(X,1)$}\\\toprule
&$x$& $y$ & $z$ & $1(s_i)$\\\midrule
$s_1$&$0$ & $0 $ & $1$ & $1$ \\
$s_2$&$0$ & $1 $ & $0$ & $1$ \\
$s_3$&$0$ & $1 $ & $2$ & $1$\\\bottomrule
\end{tabular}
\quad
\begin{tabular}{cC{3mm}C{3mm}C{3mm}C{9mm}}
\multicolumn{5}{c}{$(Y,1)$}\\\toprule
&$x$& $y$ & $z$ & $1(s_i)$\\\midrule
$s_1$&$0$ & $0 $ & $1$ & $1$ \\
$s_3$&$0$ & $1 $ & $2$ & $1$\\\bottomrule\\
\end{tabular}\quad
\begin{tabular}{cC{3mm}C{3mm}C{9mm}C{9mm}C{9mm}C{9mm}}
\multicolumn{7}{c}{$(Z, \cdot )$}\\\toprule
     &$x$& $y$ & $m(s_i)$ & $n(s_i)$ & $k(s_i)$ & $\ell(s_i)$\\\midrule
$s_1$&$0$ & $1$ & $1$ & $0$ & $1$ & $0$\\
$s_2$&$1$ & $0$ & $1$ & $1$ & $1$ & $0$ \\
$s_3$&$0$ & $0$ & $1$ & $1$ & $0$ & $1$\\\bottomrule
\end{tabular}
}
\caption{Assignments for multiteams in Examples~\ref{ex:pic4} and \ref{ex:unionclosurefails}.\label{pic45}}
\end{center}
\end{figure}
\begin{proposition}
	It is not true that $\F p(\varphi\lor\psi)\equiv\F p\varphi\lor\F p\psi$.
\end{proposition}
\begin{proof}
Let $\A$ be the multistructure over the empty vocabulary with domain $(\{0,1,2\},1)$, where $1$ is the constant $1$ multiplicity function. Then $\A\models_{(X,1)}\F{\frac{2}{3}}(x=y\lor x=z)$ but $\A\not\models_{(X,1)}\F{\frac{2}{3}}(x=y)\lor\F{\frac{2}{3}}(x=z)$, where $(X,1)$ is the multiteam depicted in the Figure \ref{pic45}.\qed
	
\end{proof}

The next simple observation states the distributivity of $\G p$ with respect to conjunction $\land$, as well as the merger of two $\F p$-operators and two $\G q$-operators, respectively.

\begin{observation} The following equivalences hold:
\begin{enumerate}
	\item $\G p(\varphi\land\psi)\equiv\G p\varphi\land\G p\psi$,
	\item $\F p(\F q\varphi)=\F{p\cdot q}\varphi$,
	\item $\G p(\G q\varphi)=\G{p\cdot q}\varphi$.
\end{enumerate}
\end{observation}

The next two examples show that both downward closure and union closure are violated by the approximate operator.


\begin{example}\label{ex:pic4}
Let $\A$ be the multistructure over the empty vocabulary with domain $(\{0,1,2\},1)$, where $1$ is the constant $1$ multiplicity function. Then $\A\models_{(X,1)}\F{\frac{1}{3}}(x=y)$ but $\A\not\models_{(Y,1)}\F{\frac{1}{3}}(x=y)$, where $(Y,1)\subseteq (X,1)$ are the multiteams depicted in the Figure \ref{pic45}.
	
\end{example}

\begin{example}\label{ex:unionclosurefails}
Let $\A$ be the multistructure over the empty vocabulary with domain $(\{0,1\},1)$, where $1$ is the constant $1$ multiplicity function. The multiteams $(Z,m)$, $(Z,n)$, $(Z,k)$, $(Z,\ell)$ are depicted in the Figure \ref{pic45}.
Now $\A\models_{(Z,k)}\G{\frac{2}{3}} (x\leq y)$ and $\A\models_{(Z,\ell)}\G{\frac{2}{3}} (x\leq y)$. However $\A\not\models_{(Z,n)} x \leq y$ and thus $\A\not\models_{(Z,m)}\G{\frac{2}{3}}(x\leq y)$ even though $(Z,k)\uplus (Z,l) = (Z,m)$.
\end{example}

\begin{proposition}
Let $\mathcal L$ be a logic and $\varphi\in\mathcal L$ a formula. Then $\F p$ preserves union closure (whereas $\G p$ does not), i.e., $\F p \varphi$ is union closed whenever $\varphi$ is.
\end{proposition}
\begin{proof}
	Let $\A$ be a multistructure and $X,Y$ be multiteams of $\A$.
	Assume that $\A\models_X \F p\varphi$ and $\A\models_Y \F p\varphi$. Then there are multiteams $X'\subseteq X$ and $Y'\subseteq Y$ such that $\lvert X'\rvert \ge p\lvert X\rvert$, $\lvert Y'\rvert\ge p\lvert Y\rvert$, and both $\A\models_{X'} \varphi$ and $\A\models_{Y'}\varphi$. Hence $\lvert X'\uplus Y'\rvert =\lvert X'\rvert +\lvert Y'\rvert \ge p\lvert X\rvert +p\lvert Y\rvert = p(\lvert X\rvert +\lvert Y\rvert ) = p\lvert X\uplus Y\rvert $ and thus $\A\models_{X\uplus Y} \F p \varphi$. \qed
\end{proof}

Yet locality holds for this logic as witnessed by the following proposition. The proof is by induction.
\begin{proposition}[Locality]\label{locality_approx}
 Let $\A$ be a multistructure, $(X,m)$ a multiteam, and $V$ be a set of variables such that $\Fr(\varphi) \sub V\sub \Dom(X)$. Then for all $\varphi\in\FO(\F p,\G p,\leq,\bbot_{\rm c},\dep(\cdot),\sub,\bot_{\rm c})$, it holds that
$\A\models_{(X,m)}\varphi$ \,iff\, $\A\models_{(X,m)\upharpoonright V}\varphi$.
\end{proposition}

\section{On the Complexity of Approximate Dependence Logic}\label{Comp}
In the following we study computational complexity of model checking in dependence logic enriched with the operator $\F p$. 
The results hold under both team and multiteam semantics. To simplify notation, we work with team semantics in this section. Analogously to \cite{2015arXiv150301144D}, our results can be seen as a first step towards a systematic classification of the syntactic fragments of approximate dependence logic for which data complexity of model-checking is tractable/intractable.

We first define the model checking problem in the context of team semantics. We consider only Boolean queries, that is we define the model checking problem for a logic $\mathcal{L}$ as follows: given a model $\A$, a team $X$ of $\A$, and a formula $\varphi$ of $\mathcal{L}$, decide whether $\A\models_X\varphi$ holds. There are three parameters to this problem: the model $\A$, the team $X$, and the formula $\varphi$. Depending on which of these parameters are fixed, a different variant of the model checking problem arises. Here we consider two of these variants: the variant with a fixed formula (this is called \emph{data complexity}), and a variant in which nothing is fixed (this is called \emph{combined complexity}).

The following two theorems reveal that already very simple formulas of approximate dependence logic witness the $\NP$-completeness of the data complexity of the logic.
\begin{theorem}\label{thm:mc-data-complexity-D(Fp)-NP-complete-AND}
 Model checking for $\F p(\dep(x,y)\land\dep(u,v))$ is $\NP$-complete.
\end{theorem}
\begin{proof}
For the lower bound we give a polynomial many-one reduction from $\threeSAT$ inspired by a similar proof of Jarmo~Kontinen \cite{2010Jarmo}. Start with a formula $\varphi=\bigwedge_{i=1}^m\bigvee_{j=1}^3 \ell_{i,j}$ where $\ell_{i,j}$ is the $j$th literal in the $i$th clause, i.e., either a variable $x$ (said of parity $0$) or its negation $\lnot x$ (of parity $1$). In the following we will construct a tuple $(X,\psi)$ from $\varphi$ such that $\varphi\in\threeSAT$ if and only if $\A\models_X\psi$. First we define the team $X$ to be the set
$$
X = \{(i,j,x,p)\mid \text{in $i$th clause the $j$th literal is the variable $x$ with parity $p$}\}.
$$

Technically the team can be seen as an encoding of the given formula. For instance the formula $\varphi=(x_1\lor\lnot x_2\lor x_3)\land(\lnot x_1\lor \lnot x_2\lor \lnot x_3)$ would yield the team
$
X = \{(1,1,x_1,0),(1,2,x_2,1),(1,3,x_3,0),(2,1,x_1,1),(2,2,x_2,1),(2,3,x_3,1)\}.
$

The formula $\psi$ is defined as
 $$
 \F{\frac{1}{3}}\bigl(\dep(\text{clause},\text{literal})\land\dep(\text{variable},\text{parity})\bigr).
 $$

Then intuitively speaking $\psi$ states that one has to decide for each clause a satisfying literal and do this consistently, i.e., the corresponding assignment has to be consistent. At first one selects exactly one third of the elements in $X$ such that for each clause a literal is chosen (i.e., \emph{clause} will determine the value of \emph{literal}). Then the \emph{parity} of each \emph{variable} is consistently chosen (i.e., \emph{variable} will determine the value of \emph{parity}). We will next formally prove that $\varphi\in\threeSAT$ if and only if $\A\models_X\psi$. 

    We first show that $\varphi\in\threeSAT\Rightarrow\A\models_X\psi$. Thus assume that $\varphi\in\threeSAT$. Let $\theta$ be an assignment such that $\theta\models\varphi$. For each $1\leq k\leq m$, let $i_k\in\{1,2,3\}$ be a number such that the literal $\ell_{k,i_k}$ in the $k$th clause of $\psi$ is satisfied by $\theta$, i.e., $\theta\models\ell_{k,i_k}$. Let $I:= \{i_1,\dots,i_k\}$. In the following we will show that $\A\models_X\psi$ holds. Define
\(    
X':= \{(k,j,v,p) \in X \mid  j = i_k\}.
\)
Clearly $\lvert X' \rvert = \frac{1}{3}\lvert X \rvert$.
%
Moreover it is easy to check that for any two $(j,i,v,p),(j',i',v',p')\in X'$
\begin{enumerate}[(a)]
	\item $j=j'$ implies $i=i'$ (the clause determines the literal) and
	\item $v=v'$ implies $p=p'$ (the variable determines the parity).
\end{enumerate}

Hence from (a) is it follows that $\A\models_{X'} \dep(\text{clause},\text{literal})$ and from (b) it follows that $\A\models_{X'} \dep(\text{variable},\text{parity})$. Since $\lvert X' \rvert = \frac{1}{3}\lvert X \rvert$, we obtain $\A\models_X\psi$.

Now turn to the direction $\A\models_X\psi\Rightarrow\varphi\in\threeSAT$ and assume that $\A\models_X\psi$. Thus there exists a team $X'\subseteq X$ such that $|X'|\ge\frac{1}{3}|X|$ and $\A\models_{X'} \dep(\text{clause},\text{literal})\land\dep(\text{variable},\text{parity})$.
Since $\A\models_{X'} \dep(\text{clause},\text{literal})$ we have that
\begin{equation}\label{eq1}
(j,i,v,p),(j',i',v',p')\in X' \text{ and } j=j' \text{ imply } i=i'.
\end{equation}
Analogously, since $\A\models_{X'} \dep(\text{variable},\text{parity})$ we have that
\begin{equation}\label{eq2}
(j,i,v,p),(j',i',v',p')\in X' \text{ and } v=v' \text{ imply } p=p'.
\end{equation}
From $\eqref{eq1}$ we can deduce that $\lvert X' \rvert \leq \frac{1}{3} \lvert X \rvert$. Since $|X'|\ge\frac{1}{3}|X|$, we obtain that $|X'|=\frac{1}{3}|X|$. This together with $\eqref{eq2}$ ensures that 
\begin{equation}\label{eq3}
\text{for each clause $j$ of $\psi$ there exits some $i,v,p$ such that $(j,i,v,p)\in X'$}.
\end{equation}
It is now easy to construct from $X'$ an assignment $\theta$ such that $\theta\models \varphi$.
%
Define
\[
\theta(v):=
\begin{cases}
1 & \text{ if $(j,i,v,0)\in X'$ for some $j,i\in \mathbb{N}$},\\
0 & \text{ if $(j,i,v,1)\in X'$ for some $j,i\in \mathbb{N}$}.
\end{cases}
\]

From $\eqref{eq2}$ it follows that $\theta$ is well-defined, whereas $\eqref{eq1}$ and \eqref{eq3} ensure that every clause of $\varphi$ is satisfied by $\theta$. Hence we have $\varphi\in\threeSAT$.

For the $\NP$ upper bound, first observe that we can simply guess a subset $X'$ of $X$ such that $|X'|\ge\frac{1}{3}|X|$. Then we just have to check whether $\A\models_{X'}\dep(\text{clause},\text{literal})\land\dep(\text{variable},\text{parity})$ holds. This can be clearly done in polynomial time.
\qed
\end{proof}

The next theorem shows that $\NP$-hard properties can be defined using very simple formulas even if the operator $\F{p}$ is restricted to appear only in front of dependence atoms. It is worth noting that the data complexity of formulas addressed in Theorem \ref{thm:mc-data-complexity-D(Fp)-NP-complete-OR-AND} without the operator $\F p$ is in NL by the results of \cite{2010Jarmo}.

\begin{theorem}\label{thm:mc-data-complexity-D(Fp)-NP-complete-OR-AND}
 Model checking for $\dep(x,y)\lor(\F p\dep(x,y)\land\dep(u,v))$ is $\NP$-complete.
\end{theorem}
\begin{proof}

The upper bound is due to the same argument as in the proof of Theorem~\ref{thm:mc-data-complexity-D(Fp)-NP-complete-AND}: use nondeterminism to tame the $\F p$ operator. The rest is just standard technique as for $\D$, see the book of V{\"a}{\"a}n{\"a}nen \cite{vaananen07}.

Now we turn to the lower bound. Here we will reduce from $\threeSAT$ through $\MaxTwoSat$, a well-known $\NP$-hard optimisation problem whose decision variant is $\NP$-complete. The problem asks given a 2CNF-formula $\varphi$ and a number $k\in\N$, if at least $k$ of the clauses of $\varphi$ can be simultaneously satisfied \cite{js76}. Garey et~al.\ describe a reduction $f$ from $\threeSAT$ to the decision variant of $\MaxTwoSat$ such that $\varphi\in\threeSAT$ iff at least $\frac{7}{10}$ of the clauses of $f(\varphi)$ can be satisfied. 

We will exploit this known reduction in the following way. The team $X$ is constructed in the same way as in the proof of Theorem~\ref{thm:mc-data-complexity-D(Fp)-NP-complete-AND}. The formula then is
  $$
 \dep(\text{clause},\text{literal})\lor(\dep(\text{clause},\text{literal})\land\F{\frac{7}{10}}\dep(\text{variable},\text{parity})).
 $$
 
Let us briefly sketch the proof as it is quite similar to the one of Theorem~\ref{thm:mc-data-complexity-D(Fp)-NP-complete-AND}. The first $\lor$ just ``removes'' the not needed half of the literals in the clauses. Then $\dep(\text{clause},\text{literal})$ takes care of that in each clause exactly one literal is chosen whereas $\F{\frac{7}{10}}\dep(\text{variable},\text{parity})$ allows us to get down to the fraction of clauses which have to be satisfied, hence have to obey the dependence atom stating that the remaining variables have to be consistently chosen, i.e., variable determines parity.
\qed
\end{proof}

Currently the $\F p$ operator is defined with respect to some value of $p\in[0,1]$. We saw that it depicts the behaviour of a \emph{ratio}. Yet we want to shortly discuss a different approach for this setting. Instead we define $\F p$ for values of $p\in\N$ hence $p$ is now a natural number with the following meaning. A team $X$ satisfies a formula $\F p\varphi$ if there exists a team $Y\subseteq X$ of size $\ge p$ such that $Y\models\varphi$---similarly for $\G p$ the meaning would be that every team $Y\subseteq X$ of size $\ge p$ satisfies $\varphi$.

Sticking to this approach would allow one to state a similar result as for Theorem~\ref{thm:mc-data-complexity-D(Fp)-NP-complete-AND} and Theorem~\ref{thm:mc-data-complexity-D(Fp)-NP-complete-OR-AND} but now for \emph{combined complexity} as follows. Here one would just explicitly state the number of rows to be removed from the team, i.e., setting $p$ to $m$ in the constructed formula in the proof of Theorem~\ref{thm:mc-data-complexity-D(Fp)-NP-complete-AND}. Regarding Theorem~\ref{thm:mc-data-complexity-D(Fp)-NP-complete-OR-AND} in this setting the formula $f(\varphi)$ increases the number of clauses by factor $10$ and therefore requires to set $p$ to $\frac{7}{10}\cdot 10\cdot m=7\cdot m$ where $m$ is the number of clauses of the given 3CNF formula $\varphi$.

%
%
%
%
%
%
%
%

\section{Conclusion}
To the best of the authors knowledge this article is the first serious approach in defining team semantics with respect to multisets for first-order dependence logic. We also initiate the study of probabilistic analogues of independence and inclusion logic. Additionally the paper provides a first step into the study of a general 
approximation operator in the team semantics framework. We show several foundational properties of these newly defined formalisms and present some first computational complexity results for approximate dependence logic ($\ADL$). For $\ADL$ we show that the introduction of approximate operators enables us to encode $\NP$-hard properties into the model checking problem (data complexity) of this logic even with only two dependence atoms, a single approximate operator, and a single conjunction. This shows how strong and elegant this kind of approximate notion really is. It is an interesting open question to study the computational properties of the analogously defined approximate inclusion logic.

Heretofore a broad field around intuitionistic logic \cite{logicAndStructure} has developed. Intuitionistic logic can be seen as classical propositional logic without the law of excluded middle. One of the main concepts here is the \emph{intuitionistic implication} $\to$. In the setting of team semantics it is defined as follows. Let $\A$ be a structure and $X$ be a team. Then $\A\models_X \varphi\to\psi$ is true if and only if for all subsets $X'\subseteq X$ it holds that $\A\models_{X'}\varphi$ implies $\A\models_{X'}\psi$.
The intuitionistic implication has been studied in the context of dependence logic, see e.g., the work of Yang \cite{yangthesis}.
An approximate variant of this operator in our setting will yield a nice resemblance to the $\G p$ operator. The slight and quite natural adjustment of intuitionistic implication to our setting is then: $\A\models_X \varphi\to_p\psi$ if and only if for all subsets $X'\subseteq X$ with $|X'|\ge p\cdot|X|$ (and $p\in[0,1]\cap\mathbb Q$) it holds that $\A\models_{X'}\varphi$ implies that $\A\models_{X'}\psi$. The operator $\G p$ can now be expressed with the help of the intuitionistic \emph{approximate} implication. One can easily verify that $\G p\varphi$ is equivalent to $\top\to_p\varphi$. 

In this article we have considered approximation in the context of  multiteam semantics when restricted to the finite. However our definitions can be generalised in a straightforward manner to deal with arbitrary cardinalities.


\section*{Acknowledgements}
The second and the third author were supported by  grants  292767, 275241 and 264917 of the Academy of Finland. The fourth author is supported by the DFG grant ME 4279/1-1. The last author was supported by the Foundations' Post Doc Pool via Jenny and Antti Wihuri Foundation.
We also thank the anonymous referees for their helpful suggestions.
 
\bibliographystyle{splncs03}
\bibliography{multisets}

\end{document}